\documentclass[letterpaper,11pt,UKenglish]{article}
\usepackage[margin=1in]{geometry}
\usepackage[ttscale=.9]{libertine}
\usepackage[T1]{fontenc}

\usepackage{amsmath,amsthm,amssymb}
\usepackage{todonotes}
\usepackage[colorlinks=true, allcolors=blue]{hyperref}
\usepackage{cleveref}
\usepackage{mathtools,thmtools}
\usepackage{thm-restate}

\usepackage{enumerate}

\DeclareMathOperator\arsinh{arsinh}

\DeclareMathOperator{\dist}{\mathrm{dist}}

\DeclareMathOperator{\cyl}{\mathrm{Cyl}}

\newcommand{\bigO}[1]{\mathcal{O}{\left(#1\right)}}
\newcommand{\bigOd}[1]{\mathcal{O}_d{\left(#1\right)}}
\newcommand{\bigOm}[1]{\Omega \left(#1\right)}

\newcommand{\bigTh}[1]{\Theta \left(#1\right)}
\newcommand{\etal}{\emph{et~al.}}

\newcommand{\Hyp}{\mathbb{H}}
\newcommand{\Reals}{\mathbb{R}}
\newcommand{\distH}[1]{\dist_{\Hyp^{#1}}}
\newcommand{\distHd}{\distH{d}}

\newcommand{\diam}{\mathrm{diam}}
\newcommand{\child}{\mathrm{child}}
\newcommand{\cheps}{\child_\eps}

\newcommand{\cT}{\mathcal{T}}
\newcommand{\cQ}{\mathcal{Q}}
\newcommand{\eps}{\varepsilon}
\newcommand{\up}{\uparrow}
\newcommand{\down}{\downarrow}
\newcommand{\zu}{z^\up}
\newcommand{\zd}{z^\down}
\newcommand{\x}{x_{\min}}
\newcommand{\xm}{x_{\max}}

\theoremstyle{plain}
\newtheorem{theorem}{Theorem}
\newtheorem{claim}[theorem]{Claim}
\newtheorem*{claim*}{Claim}

\newtheorem{lemma}[theorem]{Lemma}
\newtheorem{corollary}[theorem]{Corollary}
\newtheorem{observation}[theorem]{Observation}
\theoremstyle{definition} 
\newtheorem{definition}[theorem]{Definition}

\title{A Quadtree, a Steiner Spanner, and Approximate Nearest Neighbours in Hyperbolic Space}

\author{
S\'andor Kisfaludi-Bak \thanks{Department of Computer Science, Aalto University, Finland, \texttt{sandor.kisfaludi-bak@aalto.fi}} \and Geert {van Wordragen}\thanks{Department of Computer Science, Aalto University, Finland, \texttt{geert.vanwordragen@aalto.fi}}
}

\begin{document}

\maketitle

\begin{abstract}
    We propose a data structure in $d$-dimensional hyperbolic space that can be considered a natural counterpart to quadtrees in Euclidean spaces. Based on this data structure we propose a so-called L-order for hyperbolic point sets, which is an extension of the Z-order defined in Euclidean spaces.
    
    Using these quadtrees and the L-order we build geometric spanners. Near-linear size $(1+\eps)$-spanners do not exist in hyperbolic spaces, but we are able to create a Steiner spanner that achieves a spanning ratio of $1+\eps$ with $\mathcal O_{d,\eps}(n)$ edges, using a simple construction that can be maintained dynamically. As a corollary we also get a $(2+\eps)$-spanner (in the classical sense) of the same size, where the spanning ratio $2+\eps$ is almost optimal among spanners of subquadratic size.
    
    Finally, we show that our Steiner spanner directly provides
    a solution to the approximate nearest neighbour problem: given a point set $P$ in $d$-dimensional hyperbolic space we build the data structure in $\mathcal O_{d,\eps}(n\log n)$ time, using $\mathcal O_{d,\eps}(n)$ space. Then for any query point $q$ we can find a point $p\in P$ that is at most $1+\eps$ times farther from $q$ than its nearest neighbour in $P$ in $\mathcal O_{d,\eps}(\log n)$ time. Moreover, the data structure is dynamic and can handle point insertions and deletions with update time $\mathcal O_{d,\eps}(\log n)$.
\end{abstract}

\clearpage
\section{Introduction}
Hyperbolic spaces have special properties that set them apart from the more familiar Euclidean spaces: they expand exponentially and are tree-like, which makes them the natural choice to represent hierarchical structures.
Hyperbolic geometry has applications in several fields, including special relativity, topology, visualisation, machine learning, complex network modelling, etc.~\cite{ungar2013einstein,thurston1982three,lamping1995focus,NickelK18,krioukov2010hyperbolic}. With the growing interest in the larger scientific community, there are growing computational and graphical/visualisation needs. It is becoming increasingly important to develop basic data structures and algorithms for hyperbolic spaces. Despite clear and growing interest in machine learning~\cite{NickelK18,GaneaBH18} and the random graph/graph modelling communities~\cite{BlasiusFK16,FriedrichK18,BringmannKL19,BlasiusFKMPW22}, the data structures and algorithms in this very natural geometric setting have been largely overlooked. If we wish to process and analyse data in hyperbolic spaces in the future, the basic theory for this processing needs to be established.

Quadtrees in Euclidean spaces~\cite{FinkelB74} are among the simple early geometric data structures that have proven to be useful both in practical algorithms and in theory \cite{Aluru04,CGAA3rd,GeomAppAlg}. They form the basis of various algorithms by being able to `zoom in' efficiently. Quadtrees provide a hierarchical structure, as well as a way to think of ordering points of the plane (or higher-dimensional spaces) using the so-called Z-order curve \cite{Zorder}. They can be used as a basis for nearest neighbour algorithms~\cite{GeomAppAlg}. The first question addressed by this article is as follows.

\begin{quote}\textit{
    Q1. Is there a natural hyperbolic equivalent to Euclidean quadtrees?}
\end{quote}

Given a point set $P$ in the Euclidean plane (henceforth denoted by $\Reals^2$), a quadtree of $P$ can be defined as follows. Let $\sigma_0$ be a minimal axis-parallel square containing $P$, and let $T$ be a tree graph whose root corresponds to $\sigma_0$. Then consider a square $\sigma$ and the corresponding vertex $v_\sigma$ of $T$ where $|\sigma\cap P|\geq 2$ (starting with $\sigma = \sigma_0$). We subdivide $\sigma$ into four squares of half the side length of $\sigma$. Each smaller square is associated with a new vertex that is connected to $v_\sigma$. This procedure is repeated for each square $\sigma$  where $\sigma\cap P \geq 2$ exhaustively, until all leaves of $T$ correspond to squares that contain at most one point from $P$. The squares are called the \emph{cells} of the quadtree, and we can speak of parent/child and ancestor/descendant relationships of cells by applying the terms of the corresponding vertices in $T$. The \emph{level} of a cell or vertex is its distance to the root of $T$ along the shortest path in $T$ (i.e., the root has level $0$).

Some crucial properties of Euclidean quadtrees (slightly relaxed) are used by several algorithms. 
\begin{enumerate}
    \item \emph{Diameter doubling}. If $C'$ is a child cell of $C$, then $c_1<\diam(C')/\diam(C)<c_2$ where $0<c_1<c_2<1$ are fixed constants, and $\diam(C)$ denotes the diameter of the cell $C$.
    \item \emph{Fatness}. Each cell $C$ contains a ball that has diameter at least constant times the diameter of~$C$. Thus cells are so-called \emph{fat} objects in $\Reals^2$ \cite{fatness}.
    \item \emph{Bounded degree}. Each cell has at most $k$ children cells for some fixed constant $k$.
    \item \emph{Same-level isometry}. Cells of the same level are isometric, that is, any cell can be obtained from any other cell of the same level by using a distance-preserving transformation.
\end{enumerate}

Could the above four properties be replicated by a quadtree in the hyperbolic plane? Unfortunately this is not possible: the volume of a ball in hyperbolic space grows exponentially with its radius (thus hyperbolic spaces are not \emph{doubling spaces}). Consequently, for large cells, a cell of constant times smaller diameter than its parent will only cover a vanishingly small volume of its parent. This rules out having properties 1, 2, and 3 together. Property 4 also poses a unique challenge: while the hyperbolic plane provides many types of~\emph{tilings} one could start with, there is no transformation that would be equivalent to scaling in Euclidean spaces. This is unlike the scaling-invariance exhibited by Euclidean quadtrees. Moreover, in small neighbourhoods hyperbolic spaces are locally Euclidean, meaning that a small ball in hyperbolic space can be embedded into a Euclidean ball of the same radius with distortion infinitesimally close to 1. Thus a hyperbolic quadtree needs to operate at two different scales: at small distances we need to work with an almost-Euclidean metric, while at larger distances we need to work on a non-doubling metric.

Quadtrees in Euclidean spaces give rise to \emph{spanners} through so-called well-separated pair decompositions. Spanners are a way of representing approximate distances among the points of a point set $P$ without taking up quadratic space (i.e., without storing all $\binom{n}{2}$ pairwise distances). A spanner is a geometric graph, that is, a graph where vertices correspond to $P$, and edges are straight segments (or geodesic segments) between some pairs of points with edge lengths being equal to the distances of the underlying space. A geometric graph is called a $t$-spanner if for any pair of points $p,q\in P$ their graph distance is at most $t$ times longer than their distance in the underlying space.

Geometric spanners have a vast literature. We encourage the interested reader to look at the book of Narasimhan and Smid~\cite{narasimhan2007geometric} for an overview.
The simplest $(1+\eps)$-spanner is the greedy spanner, which considers all pairs of points sorted on increasing distance and adds an edge between a pair when their current distance in the graph is too large. This gives a spanner that is optimal in many aspects, but constructing it takes $\bigO{n^2 \log n}$ time~\cite{Bose10}.
For constant dimension $d$, a $\Theta$-graph can be constructed in $\bigO{n \log n}$ time and is a $(1+\eps)$-spanner with $\bigO{n / \eps^{d-1}}$ edges.
As shown by Le and Solomon~\cite{le2019}, this edge count is optimal.
Well-separated pair decompositions~\cite{CallahanK95} in Euclidean spaces are built upon quadtrees and naturally give rise to $(1+\eps)$-spanners that have linearly many edges.

Unfortunately well-separated pair decompositions with the required properties do not exist in hyperbolic spaces (see Corollary~\ref{cor:spanner_lower}). In fact, spanners that are not complete graphs cannot achieve a spanning ratio less than $2$ in any hyperbolic space. Luckily a spanning ratio of $1+\eps$ can be achieved if one does not insist on a geometric graph whose vertices are exactly $P$, but allows more vertices. Such spanners are called Steiner spanners, i.e., these are geometric graphs on a set $P\cup Q$ that satisfy the spanner property for pairs of points from $P$. The points of $Q$ are referred to as \emph{Steiner points}.
In Euclidean space, these Steiner points allow spanners to be sparser:
Steiner spanners can be constructed that only use $\bigO{n / \eps^{(d-1)/2} \cdot \log^2 \frac{1}{\eps}}$ edges \cite{le2019} and there is a lower bound of $\bigOm{n / \eps^{(d-1) / 2}}$ edges \cite{bhore2022}.
To our knowledge, there is no published construction of a Steiner spanner in hyperbolic space, although it was already studied in a more general setting by Krauthgamer and Lee~\cite{KrauthgamerL06}.

Quadtrees are also applicable as a basis of \emph{nearest neighbour search}: for a fixed point set $P$ and a query point $q$, can we find the point $p\in P$ that is closest to $q$ among all points in $P$? Nearest neighbour search is a well-studied and fundamental problem with numerous applications in machine learning, data analysis, and classification. Exact solutions to queries are typically only feasible in very low dimensional spaces; indeed the Euclidean methods carry over to the hyperbolic setting~\cite{NielsenN10,bogdanov2011}. However, in dimensions $3$ and above, even in the Euclidean setting, it is much more feasible to compute \emph{approximate nearest neighbours}, i.e., to find a point $p$ that is at most $1+\eps$ times farther from $q$ than the nearest neighbour of $q$.
Arya~\etal~\cite{Arya98} give a data structure for this that is constructed in $\bigO{d n \log n}$ time, requires $\bigO{dn}$ space and allows for queries in $\bigO{\lceil 1 + 6d/\eps \rceil^d \log n}$ time.
One can increase this space requirement to decrease the query time~\cite{Arya09}.
Well-separated pair decompositions~\cite{CallahanK95} and locality-sensitive orderings~\cite{doi:10.1137/19M1246493} can also be used here, with similar guarantees in query time but a worse dependence on $\eps$ and $d$ in preprocessing time and space requirements.
For higher dimensions, the exponential dependence on $d$ becomes a problem, so there one can sacrifice the optimal dependence on $n$~\cite{Indyk98}.
The second question we wish to answer is as follows.

\begin{quote}\textit{
    Q2. Is there a data structure for approximate nearest neighbour search in hyperbolic space with similar guarantees as in Euclidean space?}
\end{quote}

Our answer to both questions is affirmative.

\subparagraph*{Our contribution.} We propose a hyperbolic quadtree that satisfies properties 1, 2, as well as property 4 in case of cells of super-constant diameter. Moreover, our hyperbolic quadtree resembles a Euclidean quadtree for cells of sub-constant diameter. The authors believe that this is the first quadtree-type structure that has been proposed for hyperbolic space while considering such properties.
Based on the quadtree we are able to construct a new order and space-filling curve, named the L-order, which serves as a hyperbolic extension of the Euclidean Z-order. We show that a few hyperbolic quadtrees (and corresponding L-orders) can create a useful cover of $\Hyp^d$ in the following sense.

\begin{restatable}{theorem}{locsens}
\label{thm:locsens}
    For any $\Delta \in \Reals_+$, there is a set of at most $3d+3$ infinite hyperbolic quadtrees such that any two points $p,q \in \mathbb H^d$ with $\distHd(p,q) \leq \Delta$ are contained in a cell with diameter $\bigO{d\sqrt{d}} \cdot \distHd(p,q)$ in one of the quadtrees.
\end{restatable}

Krauthgamer and Lee~\cite{KrauthgamerL06} achieve a similar decomposition in a more general setting (on visual geodesic Gromov-hyperbolic spaces), but their construction is more implicit as it is based on a decomposition of the so-called Gromov boundary. For example, it does not immediately lend itself to an easily computable L-order.
Theorem~\ref{thm:locsens} matches the Euclidean result given by Chan, Har-Peled and Jones~\cite[Lemma 3.7]{doi:10.1137/19M1246493}. Note however that one cannot create their locality sensitive orderings in hyperbolic spaces, see Corollary~\ref{cor:spanner_lower}. We take a different route: we show that one can construct a Steiner spanner using the quadtrees of Theorem~\ref{thm:locsens}.

\begin{restatable}{theorem}{steinerspanner}\label{thm:steinerspanner}
    Let $P \in \Hyp^d$ be a given set of $n$ points and $\eps \in (0,1/2]$.
    We can construct a Steiner $(1+\eps)$-spanner for $P$ with Steiner vertex set $Q$ that has $n \cdot d^{\bigO{d}}\log(1/\eps)/\eps^d$ edges.
    The constructed spanner is bipartite with parts $P$ and $Q$, where each $p\in P$ has degree at most $d^{\bigO{d}}\log(1/\eps)/\eps^d$.
    Furthermore, a point can be inserted or deleted in $\log n \cdot d^{\bigO{d}}\log(1/\eps)/\eps^d$ time, where each insertion or deletion creates or removes at most $d^{\bigO{d}}\log(1/\eps)/\eps^d$ Steiner points and edges.
\end{restatable}

Apart from the Steiner points, this again matches the results from locality sensitive orderings \cite{doi:10.1137/19M1246493}, which are the current best results for dynamic Euclidean spanners.
With Steiner points, the current best Euclidean result has $\bigOd{\frac{n}{\eps^{(d-1)/2}} \log^2 \frac{1}{\eps}}$ edges \cite{le2019} but cannot be maintained dynamically. In the hyperbolic setting, Krauthgamer and Lee~\cite{KrauthgamerL06} gave a Steiner spanner with comparable properties, but their construction and proof is not published.

The most important step in constructing a spanner based on a cover such as the one in Theorem~\ref{thm:locsens} is to consider connections for pairs of points in a fixed quadtree cell $C$ whose distance is of the same magnitude as the cell diameter $\diam(C)$. This is typically done by subdividing this cell into cells that are of diameter $\eps \diam(C)$, and choosing hub points in each of these subcells. In our setting however the number of these subcells is unbounded, so we need a different technique. The crucial insight in our Steiner spanner construction is that the geodesics connecting pairs of points in $C$ of distance $\Omega(\diam(C))$ all go through a small region of $C$. By placing a Steiner point in this region and connecting it to the relevant points, we can approximate a large number of pairwise distances efficiently.

Our spanner has a nice bipartite structure, and with a simple trick it can be made into a $2+\eps$ spanner with the same properties (that is, a spanner without Steiner points.) Moreover, we can use this spanner to answer dynamic approximate nearest neighbour queries.

\begin{restatable}{theorem}{approxnn}\label{thm:ann}
    We can construct a data structure that uses $n \cdot d^{\bigO{d}}\log(1/\eps)/\eps^d$ space and can answer queries for a $(1+\eps)$-approximate nearest neighbour in $\log n \cdot d^{\bigO{d}}\log(1/\eps)/\eps^d$ time, and perform updates (point insertions and removals) in $\log n \cdot d^{\bigO{d}}\log(1/\eps)/\eps^d$ time.
\end{restatable}

To our knowledge, this is the first data structure for \emph{dynamic} approximate nearest neighbour search in hyperbolic space with a rigorous analysis. Krauthgamer and Lee~\cite{KrauthgamerL06} gave a static data structure\footnote{The data structure of Krauthgamer and Lee~\cite{KrauthgamerL06} achieves constant additive error, but they note that it can be extended to our current setting of a multiplicative $(1+\eps)$-approximation.} that has comparable query time but exponential storage. They also have variant with $O(n^2)$ storage and $O(\log^2 n)$ query time (each for constant $d$), but this latter construction is not published.
Our construction is comparatively simple, more efficient, and dynamic.
Once again our result matches the Euclidean result from locality sensitive orderings \cite{doi:10.1137/19M1246493}.
However, some Euclidean data structures designed specifically for approximate nearest neighbours only require $\bigO{dn}$ space and allow updates in $\bigO{d \log n}$ time \cite{Arya98}.

\subparagraph*{Further related work.}
In addition to Krauthgamer and Lee~\cite{KrauthgamerL06}, approximate nearest neighbour search in hyperbolic spaces has been studied by Wu and Charikar~\cite{wu2020nearest} and by Prokhorenkova~\etal~\cite{prokhorenkova2022graphbased}.

Wu and Charikar~\cite{wu2020nearest} describe various methods of using any algorithm for Euclidean nearest neighbour search to find exact and approximate nearest neighbours in hyperbolic space.
Their algorithms perform well in practice, but logically cannot outperform their Euclidean counterparts and they perform worse the further from Euclidean the data set is.
More concretely, they give exact and $(1+\eps)$-approximate nearest neighbours using queries to an exact algorithm for Euclidean nearest neighbours, and $\sqrt{w}(1+\eps)$-approximate nearest neighbours using queries to an algorithm for $(1+\eps)$-approximate Euclidean nearest neighbours, where $w>1$ is a parameter that influences the query time.

In general metric spaces, graph-based nearest neighbour search can be used.
These are based on a \emph{$k$-nearest neighbour graph}, a geometric graph where each point is connected to its $k$ nearest neighbours.
To answer a query, algorithms will start from some vertex in the graph and travel along edges that decrease the distance to the query point, until reaching a local minimum.
Prokhorenkova~\etal~\cite{prokhorenkova2022graphbased} show that for hyperbolic space these algorithms perform well in practice and give theoretical guarantees under some assumptions (for example, the points are uniformly distributed in a ball of radius $R$).
Their data structure has size $\bigO{n \log n \cdot M^d}$ and when $R \ll 1/\sqrt{d}$ allows queries in $\bigO{n^{1/d} \cdot M^d}$ time, where $M$ is a constant related to the approximation factor.
When $R \gg \log d$, the query time becomes $\bigO{n^{1/d} \cdot M^d R / e^{R(d-1)/d}}$:
queries become faster when points are further apart.

Some decompositions similar to our quadtree have been considered in the context of hyperbolic random graphs.
Von Looz, Meyerhenke and Prutkin~\cite{looz2015} introduced a \emph{polar quadtree}.
For this, they represent points in the hyperbolic plane by polar coordinates:
the angle $\phi$ and distance $r$ w.r.t.\ a fixed direction and origin.
Each cell of the quadtree is then of the form $[\phi_{\min}, \phi_{\max}) \times [r_{\min},r_{\max})$.
As in a Euclidean quadtree, a cell can be split into four children.
Cells of the same level can be set to have the same area, but their diameters will vary significantly, which also means they are not isometric and can get arbitrarily thin.

Various papers about hyperbolic random graphs~\cite{bode2015,fountoulakis2018,müller2019} use a discretisation similar to the binary tiling (which is introduced in the next section and forms the basis of our quadtree), but based on the polar coordinate system instead of the half-space model.

\subparagraph*{Organisation of the paper.}
Section~\ref{sec:preliminaries} introduces the properties and concepts from hyperbolic space used in this paper, as well as some notation.
Section~\ref{sec:quadtree} gives the construction of our quadtree and proves several of its properties.
In particular, it also shows that a technique similar to shifting can be used, and that we can compute a hyperbolic equivalent to the Z-order we refer to as the L-order.
Section~\ref{sec:approx} first gives a lower bound result for spanners in hyperbolic space.
After that, it uses the proven quadtree properties for the construction of a Steiner $(1+\eps)$-spanner and shows that it can be maintained dynamically, as well as used for approximate nearest neighbour search.
Section~\ref{sec:constapprox} shows that constant-approximate nearest neighbours can already be found using only shifting and the L-order.
Section~\ref{sec:quadtreeproof} gives a computation-heavy proof for some fundamental properties of our quadtree.
Finally, we end with a conclusion in Section~\ref{sec:conclusion}.

\section{Preliminaries}\label{sec:preliminaries}
The reader should be able to grasp the overall structure of our quadtree and some of the consequences of this structure even if they are new to hyperbolic geometry; we suggest thoroughly understanding binary tilings (see below), the basic shapes (shortest paths, hyperplanes), and probing the distance formula to get some basic intuition. A more thorough reader will need some familiarity with the basics of hyperbolic geometry and trigonometry, as well as the Poincar\'e half-space model. For more background on hyperbolic geometry, please see~\cite{cannon1997hyperbolic} and the textbooks~\cite{iversen1992hyperbolic,thurston97three,benedetti1992lectures}.
Note that our results are presented in the half-space model, but are in fact model-independent. Our algorithms are also presented for point sets whose coordinates are given in the half-space model. Apart from numerical challenges ---something we will not tackle in this article---, working in the half-space model is not restrictive as conversion between various models of hyperbolic geometry is straightforward.

Let $\Hyp^d$ denote the $d$-dimensional hyperbolic space of sectional curvature $-1$. 
We denote points as $(x,z)$ for $x \in \Reals^{d-1}$ and $z \in \Reals^+$, with the distance
\[
    \distHd((x,z), (x',z')) = 2\arsinh\left(\frac12 \sqrt{\frac{\|x-x'\|^2 + (z - z')^2}{zz'}} \right),
\]
where $\|x\|$ refers to the $(d-1)$-dimensional Euclidean norm of $x$.
When $x=x'$ this reduces to $\left| \ln\left( \frac{z}{z'} \right) \right|$.

For the rest of the paper, we fix a particular half-space model, and describe our results in this model. We will think of the $z$ direction as going ``up'', and the $d-1$ other axes are going ``sideways''.

The transformations $T_{\sigma,\tau}(x,z) = \sigma \cdot (x + \tau, z)$, where we translate $x$ with a vector $\tau \in \Reals^{d-1}$ and then scale all coordinates by $\sigma \in \Reals^+$, are isometric.
This can be verified by applying $T_{\sigma,\tau}$ to both arguments in the distance formula.

Consider now all the transformations $T_{\sigma,\tau}$ where $\sigma = 2^k$ for some integer $k$ and $\tau$ is an integer vector. One can observe that acting upon the Euclidean unit cube with corners $(0,\dots,0,1)$ and $(1,\dots,1,2)$ these transformations together create a tiling of the half-space model with isometric tiles. This tiling has been named the \emph{binary tiling}, and it was introduced by Böröczky~\cite{Boro}, see also~\cite{cannon1997hyperbolic}. The binary tiling is the basis of our quadtree construction. The $2$-dimensional binary tiling is illustrated in Figure~\ref{fig:binary}.

\begin{figure}[ht]
    \centering
    \includegraphics[scale=.5]{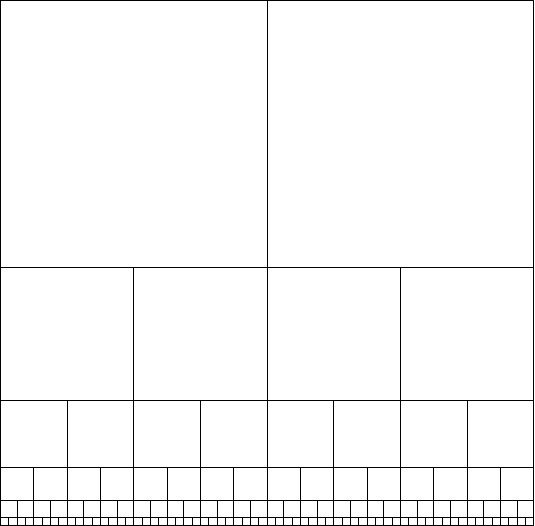}
    \caption{A portion of the binary tiling in the half-plane model. The reader may wish to think of the hyperbolic plane itself as the surface one gets by gluing rubber squares \emph{of the same size}  along their sides as shown in the picture, resulting in a smooth and homogeneous infinite surface, which is the hyperbolic plane itself.}
    \label{fig:binary}
\end{figure}

In the half-space model, the shortest path (\emph{geodesic}) between two points only matches the Euclidean line segment between them when this segment is vertical.
Otherwise, the geodesic is an arc of a Euclidean circle that has its origin at $z=0$ (this means it is similar to the shortest path in the dual graph of the binary tiling).
For hyperplanes the situation is similar:
vertical Euclidean hyperplanes are hyperbolic hyperplanes, but the other hyperbolic hyperplanes are Euclidean spheres with their origin somewhere on the plane $z=0$.
This paper will often also use horizontal Euclidean hyperplanes, so it is worth noting that these are \emph{not} hyperplanes in the hyperbolic sense (for example, they can be intersected twice by the same geodesic), but a different shape called a horosphere.
Finally, the half-space model is conformal: Euclidean and hyperbolic angles are the same. Additionally it means that any Euclidean sphere that does not intersect $z=0$ is also a hyperbolic sphere, but its hyperbolic origin will lie lower than its Euclidean origin.

\section{A hyperbolic quadtree}\label{sec:quadtree}
The Euclidean quadtree is a tree whose vertices are associated with axis-parallel hypercubes.
Correspondingly, our hyperbolic quadtree will be a tree whose vertices are associated with so-called \emph{cube-based horoboxes}.
For this paper, it will be useful to define an \emph{axis-parallel horobox} as the shape that corresponds to a Euclidean axis-parallel box in a fixed half-space model.
Such a horobox $B$ can be defined by its corner points $(\x(B),\zd(B))$ and $(\xm(B), \zu(B))$, where one is minimal and the other maximal in all coordinates.
Notice that a horobox is bounded by 2 horospheres and $2(d-1)$ hyperplanes.
The cube-based horobox is a special axis-parallel horobox, defined as follows.

\begin{definition}
    In a fixed half-space model, a \emph{cube-based horobox} $C$ is an axis-parallel horobox with $\frac{\zu(C)}{\zd(C)} = 2^{h}$ and $\frac{\xm(C) - \x(C)}{\zd(C)} = (w,\dots,w)$, where $w=w(C)$ is called the \emph{width} of $C$ and $h=h(C)$ is called the \emph{height} of $C$.
\end{definition}

It is worth noting that for $\zd(C)=1$, the width corresponds to the Euclidean width of $C$.
On top of that, defining the width and height in this way ensures that two cube-based horoboxes $C',C$ with the same width and height are congruent to one another:
setting $\sigma = \frac{\zd(C')}{\zd(C)}$ and $\tau = \x(C') / \sigma - \x(C)$ gives $T_{\sigma,\tau}(C) = C'$.

\begin{lemma}
\label{lem:diam}
    The cube-based horobox $C$ has diameter
    \[ \diam(C) = \begin{cases}
        2\arsinh\left(\frac12 w(C) \sqrt{d-1}\right) & \text{if }w(C) \geq \sqrt{\frac{2^{h(C)} - 1}{d-1}}, \\
        2\arsinh\left(\frac12 \sqrt{\frac{(d-1)w(C)^2 + (2^{h(C)}-1)^2}{2^{h(C)}}} \right)  & \text{otherwise.}
    \end{cases}\]
\end{lemma}
\begin{proof}
    Let $w=w(C)$ and $h=h(C)$.
    The distance $\distHd((x,z), (x',z'))$ is monotone increasing in $\|x - x'\|$.
    Assuming $z' \leq z$, the distance is also monotone decreasing in $z'$.
    This means the diameter is w.l.o.g.\ given by $\distHd\left((0,z), ((w, \dots, w),1)\right) = 2\arsinh\left( \frac{1}{2} \sqrt{\frac{(d-1) w^2 + (z-1)^2}{z}} \right)$, for some $1 \leq z \leq 2^{h}$.

    In this interval the function is convex in $z$, so the maximum is attained at either $z=1$ or $z=2^{h}$.
    Thus it is either $2\arsinh\left(\frac12 w \sqrt{d-1}\right)$ or $2\arsinh\left(\frac12 \sqrt{\frac{(d-1)w^2 + (2^{h}-1)^2}{2^{h}}} \right)$.
    These are equal when $w = \sqrt{\frac{2^{h} - 1}{d-1}}$.
    Both are continuous, so it suffices to check two combinations of values to see that $z=1$ gives the largest value if and only if $w \geq \sqrt{\frac{2^{h} - 1}{d-1}}$.
\end{proof}

\subsection{Hyperbolic quadtree construction.} \label{sec:construct}
One property of hyperbolic space that makes quadtrees more complicated is that it behaves differently at different scales:
in small enough neighbourhoods the distortion compared to Euclidean space becomes negligible,
but at larger scales the hyperbolic nature becomes more and more pronounced.
This means that the quadtree also has to work differently at different scales.

For a point set $P \subset \Hyp^d$, the hyperbolic quadtree $\cQ(P)$ is a graph whose vertices are regions of hyperbolic space.
We can construct $\cQ(P)$ as follows.
First, we find the Euclidean minimum bounding box of $P$ in the half-space model. 
From this we can get a minimum bounding cube-based horobox $C_\text{bound}$ where $\zd(C_\text{bound}) = \min_{p\in P} z(p)$ (i.e., we shift the horobox up as much as possible). 

In case of $d=2$, our goal is to ensure that quadtree cells of level $\ell\geq 0$ correspond to horoboxes whose vertices come from a fixed binary tiling. Note that unlike the Euclidean setting where the quadtree levels are usually defined based on a point set, our quadtree has its levels defined based on the absolute size of the cells. Higher levels will correspond to larger cells, and level $0$ will act as the transition between the hyperbolic and Euclidean ways of splitting cells.
The levels can be constructed starting at level $\ell=0$, where cells are the tiles of a binary tiling.
The binary tiling is closely related to binary trees:
we get a binary tree from the tiling by making a graph where each vertex corresponds to a horobox and edges correspond to them being vertically adjacent.
When we have a binary tree, we can partition it into a small number of identical subgraphs by `cutting' at half its depth, see Figure~\ref{fig:binarycut}(ii).
This gives a natural way to split cells of level $\ell\geq 1$ into $1+2^\ell$ isomorphic cells of level $\ell - 1$, one corresponding to the `top' part of the binary tree, and the rest corresponding to the subtrees defined by the vertices at depth $\ell$. When splitting cells of level $\ell \leq 0$, we are already in a setting where the distortion is very small compared to the Euclidean setting; here we simply use Euclidean dissection into four smaller cells, each with the same Euclidean width and hyperbolic height.
Figure~\ref{fig:quadtree}(iii) shows an example of the resulting hyperbolic quadtree.

\begin{figure}[t]
    \centering
    \includegraphics{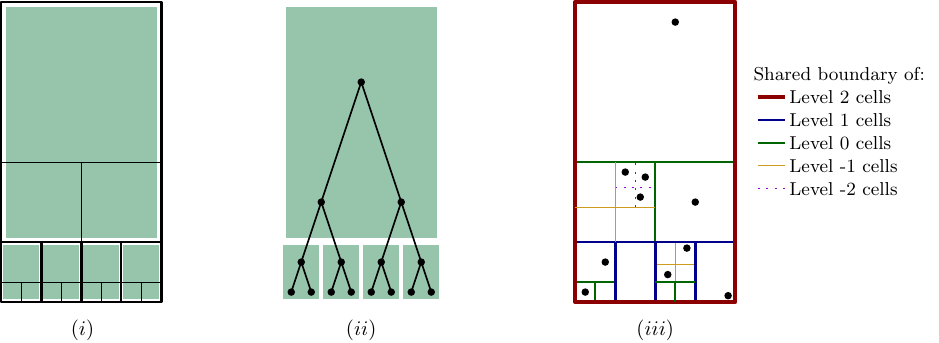}
    \caption{(i) A quadtree cell of level $2$ is split into $5$ cells of level $1$.
    (ii) The binary tree of depth $4$ is split into $5$ isomorphic binary trees of depth $2$.
    (iii) A hyperbolic quadtree where the root (red) is a level $2$ cell, which is split into $5$ isometric cells of level $1$ that are separated by blue Euclidean segments.}
    \label{fig:binarycut}
    \label{fig:quadtree}
\end{figure}

We can also define the quadtree for general $d$.
First, we define its root cell $C_\text{root}$ with $\x(C_\text{root}) = \x(C_\text{bound})$ and $\zd(C_\text{root}) = \zd(C_\text{bound})$, and furthermore
\begin{itemize}
    \item If\footnote{Notice that the $1/\sqrt{d-1}$ terms in this definition simplify the diameter formula of quadtree cells, making them independent of $d$: see Lemma~\ref{lem:cells}.} $w(C_\text{bound}) \leq \frac{1}{\sqrt{d-1}}$ and $h(C_\text{bound}) \leq 1$, we find the smallest integer $\ell$ such that $w(C_\text{bound}) \leq \frac{2^\ell}{\sqrt{d-1}}$ and $h(C_\text{bound}) \leq 2^\ell$, then let $w(C_\text{root}) = \frac{2^\ell}{\sqrt{d-1}}$ and $h(C_\text{root}) = 2^\ell$.
    \item Otherwise, we find the smallest integer $\ell$ such that $w(C_\text{bound}) \leq \frac{2^{2^\ell-1}}{\sqrt{d-1}}$ and $h(C_\text{bound}) \leq 2^\ell$, then let $w(C_\text{root}) = \frac{2^{2^\ell-1}}{\sqrt{d-1}}$ and $h(C_\text{root}) = 2^\ell$. 
\end{itemize}

We then subdivide cells to get their children, but unlike with the Euclidean quadtree this subdivision depends on the size of the cell.
If we have a cell $C$ with $h(C) \leq 1$, then we split it into $2^d$ smaller ones using the
axis-parallel Euclidean hyperplanes through $\left( \frac{\x(C) + \xm(C)}{2}, \sqrt{\zd(C) \zu(C)} \right)$.
For larger cells, we also use the Euclidean hyperplane $z=\sqrt{\zd(C) \zu(C)}$.
This gives two horoboxes with height $h(C)/2$, where the top one has width $w(C) / 2^{\frac{h(C)}{2}}$ but the bottom one still $w(C)$.
Thus, we also split the bottom horobox into a grid of $2^{\frac{h(C)}{2}(d-1)}$ horoboxes of width $w(C) / 2^{\frac{h(C)}{2}}$ so that in total we have $2^{\frac{h(C)}{2}(d-1)} + 1$ cells of the same size.

\begin{lemma}\label{lem:cells}
    At any level $\ell$, cells are cube-based horoboxes with height $2^\ell$.
    For $\ell \geq 0$, the width is $\frac{2^{2^\ell - 1}}{\sqrt{d-1}}$ and the diameter is $2\arsinh(2^{2^\ell - 2})$.
    For $\ell < 0$, the width is $\frac{\alpha \cdot 2^\ell}{\sqrt{d-1}}$ and the diameter is $2\arsinh\left(\frac12 \sqrt{\frac{\alpha^2 \cdot 4^\ell + (2^{2^\ell}-1)^2}{2^{2^\ell}}} \right)$, where $\alpha \in \left(\frac{1}{2},1\right]$ is a cell-specific value. Moreover, if a cell $C$ of level $\ell$ has corresponding value $\alpha$ and a child cell $C'$ of $C$ has corresponding value $\alpha'$, then $\alpha'/\alpha \in \{1,2^{-2^{\ell-1}}\}$.
\end{lemma}
The final claim in this lemma shows that sibling cells will become extremely close to being isometric, as $2^{-2^{\ell-1}}$ rapidly converges to $1$ as $\ell$ goes to $-\infty$.
\begin{proof}
    The statement for $\ell \geq 0$ follows directly from the construction and Lemma~\ref{lem:diam}.
    At $\ell \leq 0$, a cell $C$ gets split into $2^d$ cells that all have height $h(C)/2$, where the lower $2^{d-1}$ have width $\frac{w(C)}{2}$ and the upper $2^{d-1}$ have width $\frac{w(C)}{2} / 2^{\frac{h(C)}{2}}$.
    At level 0 the height is $1$ and the width $\frac{1}{\sqrt{d-1}}$, so cells at level $\ell < 0$ have height $2^\ell$ and width $\frac{\alpha \cdot 2^\ell}{\sqrt{d-1}}$, for a yet to be determined value $\alpha$.
    Its lower child cells have width $\frac{w(C)}{2} = \frac{\alpha \cdot 2^{\ell-1}}{\sqrt{d-1}}$, while the upper children have width $\frac{w(C)}{2} / 2^{\frac{h(C)}{2}} = \frac{\alpha \cdot 2^{\ell-1}}{2^{2^{\ell-1}} \sqrt{d-1}}$.
    Thus the width of the child cells follows the same formula, with $\alpha$ the same or replaced by $\alpha \cdot 2^{-2^{\ell-1}}$.
    At $\ell = 0$ we have $\alpha = 1$, thus it is also the highest value it can take for other cells.
    The lowest possible value of $\alpha$ at level $\ell$ is given by $\prod_{i=\ell+1}^0 2^{-2^{i-1}} = 2^{-\sum_{i=\ell}^{-1} 2^i} = 2^{2^\ell-1}>\frac{1}{2}$.
\end{proof}

\begin{restatable}{theorem}{quadtree}\label{thm:quadtree}
    The hyperbolic quadtree of $P \subset \Hyp^d$ has the following properties:
    \begin{enumerate}[(i)]
        \item If $C'$ is a child cell of $C$, then $0.42 < \diam(C')/\diam(C) < 0.561$.
        \item If $C$ is a cell at level $\ell$, then $\diam(C) = \bigTh{2^\ell}$.
        \item Cells are $\Omega(1/\sqrt{d})$-fat.
        \item A quadtree cell $C$ has $\max(2^d,2^{\bigO{ d \cdot \diam(C) }})$ children; in particular, the root has $\max(2^d,d^{\bigO{ d \cdot \diam(P) }})$ children.
        \item Cells of the same level $\ell \geq 0$ are isometric, and cells of level $\ell<0$  are cube-based horoboxes with the same height whose width differs by less than a factor two.
    \end{enumerate}
\end{restatable}

The proof of Theorem~\ref{thm:quadtree} is straightforward but calculation-heavy. Its proof can be found in Section~\ref{sec:quadtreeproof}.

\subsection{Covering with hyperbolic quadtrees}\label{sec:covering}
Euclidean quadtrees are useful in computing nearest neighbours and other related problems because of a particular distance property: there is a small collection of quadtrees one can define such that any pair of points at distance $\delta$ will be contained in a cell of diameter $\bigO{\delta}$ in one of the quadtrees. Moreover, the quadtrees can be generated by simply taking one quadtree and translating (\emph{shifting}) it with different vectors.
We will prove an analogous property for our hyperbolic quadtrees, though our ``shifts'' are the transformations $T_{\sigma,\tau}$ instead of translations. Let us first introduce an infinite quadtree.

Consider the binary tiling that contains the Euclidean hypercube with opposite corners $(0,\dots,0,1)$ and $(1,\dots,1,2)$. This tiling forms level $0$ of the infinite quadtree $\cQ^d_\infty$. Then each level $\ell<0$ is defined by subdividing these cells according to the construction in Section~\ref{sec:construct}. For $\ell>0$ we define the level $\ell$ cells by unifying $2^{(d-1)\ell}+1$ cells of level $\ell-1$ into a horobox, doing the splitting described in Section~\ref{sec:construct} in reverse. We do the unification in such a way that the Euclidean hyperplanes $x_i=0$ for $i=1,\dots,d-1$ as well as $z=1$ remain cell boundaries at each level $\ell$. More formally, cells of level $\ell$ are the cube-based horoboxes $C$ where
\[
    h(C) = 2^\ell,\quad
    \zd(C) = 2^{b \cdot h(C)},\quad
    w(C) = \frac{2^{2^\ell-1}}{\sqrt{d-1}},\quad
    \x(C) = a \cdot w(C) \cdot \zd(C),
\]
for each $(a,b)$ where $a\in \mathbb{Z}^{d-1}$ and $b\in \mathbb{Z}$.
As a result we get the infinite quadtree $\cQ^d_\infty$ where for each level $\ell$ the cells of the quadtree define a subdivision of $\Hyp^d$.

For $x,y \in \Reals^+$ we define $x \bmod y = x - y \lfloor x/y \rfloor$.
Chan \etal~\cite{doi:10.1137/19M1246493} observed that shifting a quadtree by certain special vectors results in a useful shifts also for levels with smaller cells. The following lemma was used to define these useful shifts.
\begin{lemma}[Chan \etal~\cite{doi:10.1137/19M1246493}]\label{lem:mod}
    Let $n > 1$ be a positive odd integer, and consider the set
    \[ X = \{ i/n \mid i=0, \dots, n-1 \}. \]
    Then, for any $\alpha = 2^{-\ell}$, where $\ell \geq 0$ is an integer, we have that
    \[ X \bmod \alpha = \{ i/n \bmod \alpha \mid i=0, \dots, n-1 \} \]
    is equal to the set $\alpha X = \{ \alpha i/n \mid i=0, \dots, n-1 \}$.
\end{lemma}

We will look at the mapping $\tilde\pi_z(p) = \log z(p)$, where $\log$ denotes the base-$2$ logarithm.
Applying $\tilde\pi_z$ to a cell $C$ gives an interval in $\Reals$, which we will refer to as the \emph{$z$-range} of $C$.
We can also apply $\tilde\pi_z$ to the quadtree as a whole and merge all nodes that have the same $z$-range.
A cell $C$ has $z$-range $[\log \zd(C), \log \zu(C))$ and by construction its children have $z$-range $\left[\log \zd(C), \frac{\log\zd(C) + \log\zu(C)}{2} \right)$ or $\left[\frac{\log\zd(C) + \log\zu(C)}{2}, \log \zu(C) \right)$.
Thus, we have a tree of intervals where splitting an interval in the middle gives its two children.
This is exactly the structure of a one-dimensional Euclidean quadtree.
Additionally, $\tilde\pi_z(T_{\sigma,\tau}(p)) = \log\sigma + \tilde\pi_z(p)$, meaning that shifts are an isometry we can apply to this one-dimensional quadtree.
This lets us use the following lemma:

\begin{lemma}\label{lem:1Dshift}
    Let $\cT$ be a one-dimensional Euclidean quadtree whose largest cell is $[0,2)$.
    For any two points $p,q \in [0,1)$, there is a shift $\sigma \in \{0, \frac13, \frac23 \}$ such that when added to the quadtree $p + \sigma$ and $q + \sigma$ are contained in a cell of $\cT$ with length $< 3|p-q|$ and one of the points is in the lower $\frac13$ of the cell.
\end{lemma}
\begin{proof}
    Let $\alpha = 2^{-\ell}$ for some $\ell \in \mathbb N_0$ such that $\alpha/2 < |p - q| \leq \alpha$.
    We apply Lemma~\ref{lem:mod} for $n=3$ and $\alpha$: as a consequence, a cell of size $\alpha$ can be thought of as being shifted by one of $\{ 0, \frac{\alpha}{3}, \frac{2\alpha}{3} \}$.
    These shifts divide the space into intervals of length $\frac{\alpha}{3}$ where any three adjacent intervals will make up the cell of the quadtree under some shift.
    Thus, $p$ and $q$ will always be contained in a cell with length $\alpha$ when $|p-q| \leq \frac{2\alpha}{3}$.
    When $|p-q| > \frac{2\alpha}{3}$, we still know $|p - q| \leq \alpha \leq \frac{4\alpha}{3}$ so they must still be contained in a cell $C$ with length $2\alpha$.
    Thus, in the worst case, $C$ has length $2\alpha < 3|p - q|$. If $C$ has $p$ and $q$ both in its higher $\frac{2}{3}$, then the quadtree shifted by $\frac{\alpha}{3}$ or $\frac{2\alpha}{3}$ will contain a cell $C'$ of the same level as $C$ with the desired property.
\end{proof}

Now we will consider all cells with some given $z$-range and apply the mapping $\tilde\pi_x(p) = x(p) \sqrt{d-1}$.
This produces a grid in $\Reals^{d-1}$ where the cells have some power of two as the side length.
We again have access to shifts as an isometry, because $\tilde\pi_x(T_{\sigma,\tau}(p)) = \tau \sigma \sqrt{d-1} + \sigma \cdot \tilde\pi_x(p)$.
This combination lets us use a lemma by Chan~\cite{Chan1998} about \emph{$\delta$-centrality}.
We say that a point is \emph{$\delta$-central} in an axis-parallel hypercube of side length $r$ when its Euclidean distance to the cell boundary is at least $\delta r$.

\begin{lemma}[Chan~\cite{Chan1998}]\label{lem:central}
    Suppose $d$ is even.
    Let $v^{(j)} = (j/(d+1), \dots, j/(d+1)) \in \Reals^d$.
    For any point $p \in \Reals^d$ and $r = 2^{-\ell}\ (\ell \in \mathbb N)$, there exists $j \in \{0, 1, \dots, d\}$ such that $p + v^{(j)}$ is $(1/(2d+2))$-central in its hypercube with side length $r$.
\end{lemma}

We need one final observation.

\begin{observation}\label{obs:planepoint}
    The distance from $(x, z) \in \Hyp^d$ to the hyperplane $x_1=0$ is $\arsinh \frac{|x_1|}{z}$.
\end{observation}
\begin{proof}
    Fix a point $p$ and a hyperbolic hyperplane $h$ in $\Hyp^d$.
    The \emph{reflection} of $p$ on $h$ is a point $p'$ such that the geodesic $pp'$ is perpendicular to $h$, and the midpoint $t$ of $pp'$ is on $h$. We call $t$ the \emph{hyperbolic projection} of $p$ onto $h$.
    Reflecting $p = (x,z)$ on the hyperplane $x_1=0$ gives a point $p' = (x',z)$ where $x'_1 = -x_1$ but still $x'_i = x_i$ for all other $i$.
    Let $t$ be the hyperbolic projection of $p$ onto $x_1 = 0$.
    By definition, $t$ is the midpoint of $pp'$, so $\distHd(p, t) = \frac12 \distHd(p, p') = \arsinh{\frac{|x_1|}{z}}$.
\end{proof}

Using this observation and the behaviour under projections we can now prove Theorem~\ref{thm:locsens}.

\locsens*
\begin{proof}
    Let $L \in \mathbb N$ be such that the sphere with radius $\Delta$ can be covered by a cube-based horobox with width $W = \frac{2^{2^L-1}}{\sqrt{d-1}}$ and height $H = 2^L$.
    Let $D$ be $d$ rounded down to the nearest even number, i.e.\ $D = 2\lfloor d/2 \rfloor$.
    For each combination of $i \in \{0,1,2\}$ and $j \in \{0, \dots, D\}$ we define the quadtree $\cQ_{ij} = T_{\sigma_i,\tau_j}^{-1}(\cQ^d_\infty)$, with $\sigma_i = 2^{H \cdot i/3}$ and $\tau_j = (W \cdot j/(D+1), \dots, W \cdot j/(D+1))$.
    In the proof we will apply $T_{\sigma_i,\tau_j}$ to all points in $\cQ_{ij}$ instead of transforming the cells of the quadtree itself, but this has the same effect.
    
    Let $p,q$ be arbitrary points with $\distHd(p,q) \leq \Delta$.
    Lemma~\ref{lem:1Dshift} gives a $\sigma_i$ such that $\tilde\pi_z(T_{\sigma_i,\tau}(p))$ and $\tilde\pi_z(T_{\sigma_i,\tau}(q))$ are contained in a one-dimensional cell with length less than $3|\tilde\pi_z(p) - \tilde\pi_z(q)|$, where $\tau$ can have any value because it is discarded in the projection.
    This corresponds to $T_{\sigma_i,\tau}(p)$ and $T_{\sigma_i,\tau}(q)$ being contained in a $z$-range with the same length.
    One of $\tilde\pi_z(T_{\sigma_i,\tau}(p))$ and $\tilde\pi_z(T_{\sigma_i,\tau}(q))$ is in the lower $\frac13$ of this $z$-range; without loss of generality we assume this is $p$.

    We now consider the $d$-dimensional level-$\ell$ cells where $\ell$ is large enough to assure the above holds.
    According to Lemma~\ref{lem:central}, there is also a $\tau_j$ such that $\tilde\pi_x(T_{\sigma_i,\tau_j}(p))$ is $\frac{1}{2D+2}$-central in its $(d-1)$-dimensional level-$\ell$ cell $\tilde\pi_x(C)$.
    Let $p' = T_{\sigma_i,\tau_j}(p)$ and $q' = T_{\sigma_i,\tau_j}(q)$.
    If $q'$ lies in a different cell with the same $z$-range, then $\distHd(p',q')$ must be greater than the distance between $p'$ and the bounding hyperplanes of $C$.
    Without loss of generality, we can assume that $\x(C) = 0$, $\zd(C) = 1$ and consider the hyperplane $H$ given by $x_1=0$.
    By Observation~\ref{obs:planepoint} $d(p', H) = \arsinh\frac{|x_1(p')|}{z(p')}$.
    We have $x_1(p') \geq \frac{w(C)}{2D+2}$ and $z(p') \leq 2^{h(C)/3}$, because $\tilde\pi_z(p')$ is in the lower $\frac13$ of $\tilde\pi_z(C)$.
    Thus, if $\distHd(p',q') \leq \arsinh\left(\frac{w(C)}{(2D+2)2^{h(C)/3}}\right)$, then also $q' \in C$.
    For the cell $C$ at the lowest level $\ell$ where $p',q' \in C$, define $f_d(\ell) = \frac{\diam(C)}{\distHd(p',q')}$.
    
    We first consider $\ell \geq 0$, where $\diam(C_p) = 2\arsinh\left(2^{2^\ell-2} \right)$ and $w(C_p) = \frac{2^{2^\ell-1}}{\sqrt{d-1}}$ by Lemma~\ref{lem:cells}, so
    \[
        f_d(\ell)
        \leq \frac{ 2\arsinh\left(2^{2^\ell-2} \right) }{ \arsinh\left(\frac{\sqrt[3]{2^{2^\ell}}}{4(D+1){\sqrt{d-1}}}\right) }
        \leq \frac{ 2\ln\left(2^{2^\ell-1} \right) + 2\arsinh{\frac12} }{ \arsinh\left(\frac{\sqrt[3]{2^{2^\ell}}}{4(D+1){\sqrt{d-1}}}\right) }.
    \]
    To bound this further, we let $\alpha = \frac{\sqrt[3]{2^{2^\ell}}}{4(D+1){\sqrt{d-1}}}$ and first assume $\alpha \geq 2$.
    This means $2^{2^\ell} = \bigO{\alpha^3 d^{\frac{9}{2}}}$.
    We always have $\arsinh\alpha \geq \ln 2\alpha > 0$, so
    \[
        f_d(\ell)
        = \bigO{\frac{\log\left( \alpha^3 d^{\frac{9}{2}} \right)}{\arsinh\alpha}}
        = \bigO{\frac{\log\alpha + \log d}{\log\alpha}}
        = \bigO{\log d}.
    \]
    When we instead assume $\alpha < 2$, then $\arsinh\alpha \geq \frac{\arsinh2}{2} \cdot \alpha$ and
    \[
        f_d(\ell)
        = \bigO{\frac{2 \log\left(2^{2^\ell-2} \right)}{\frac{\sqrt[3]{2^{2^\ell}}}{4(D+1){\sqrt{d-1}}}}}
        = \bigO{\frac{2^\ell d \sqrt d}{\sqrt[3]{2^{2^\ell}}}}
        = \bigO{d \sqrt d}.
    \]
    
    For $\ell < 0$, Lemma~\ref{lem:diam} implies that the width is at least $\frac{\frac{1}{2} \cdot 2^\ell}{\sqrt{d-1}}$, the height is at most $1$, and the diameter is at most $2\arsinh\left(\frac12 \sqrt{\frac{4^\ell + (2^{2^\ell}-1)^2}{2^{2^\ell}}} \right)$.
    Finally, we use that $\arsinh x \leq x$ for any $x$ and $\arsinh{x} = \bigOm{x}$ for bounded $x$:
    \[
        f_d(\ell)
        \leq \frac{2\arsinh\left(\frac12 \sqrt{\frac{4^\ell + (2^{2^\ell}-1)^2}{2^{2^\ell}}} \right)}{\arsinh\left(\frac{\frac{1}{2} \cdot 2^\ell}{(2D+2)\sqrt[3]{2}\sqrt{d-1}}\right)}
        = \bigO{\frac{\sqrt{\frac{4^\ell + (2^{2^\ell}-1)^2}{2^{2^\ell}}}}{\frac{2^\ell}{d\sqrt{d}}} }
        = \bigO{d \sqrt d}.\qedhere
    \]
\end{proof}

\subsection{L-order}
When we have the Euclidean quadtree for a set of points, we can do a depth-first traversal of the tree and note in which order the points are visited.
This gives rise to the Z-order.
As it turns out, adding or removing points does not change the Z-order and for a pair of points we can determine which comes first without ever constructing a quadtree.
The only thing to specify is which infinite quadtree their quadtree would be a subset of, because a differently shifted quadtree can give different results.

We can do the same to get the L-order from a hyperbolic quadtree.
Here, we first need to define how exactly we do the depth-first traversal.
For levels $\ell > 0$, we first visit the top child and then visit the bottom children in Z-order.
For lower levels, the split is the same as for Euclidean quadtrees so we visit the children in the same order as the Z-order.

\begin{lemma}\label{lem:lorder}
    For two points $p,p' \in \Hyp^d$, we can check which comes first in the L-order for $\cQ^d_\infty$ by using $\bigO{d}$ floor, logarithm, bitwise logical and standard arithmetic operations.
\end{lemma}
\begin{proof}
    To compare $p$ and $p'$, we first define $\tilde\pi(p)=(\tilde\pi_x(p), \tilde\pi_z(p))$, then let $(x,z)=\tilde\pi(p)$ and $(x',z')=\tilde\pi(p')$ for convenience.
    We then check if both $\lfloor x/2^z \rfloor = \lfloor x'/2^{z'} \rfloor$ and $\lfloor z \rfloor = \lfloor z' \rfloor$.
    If that is the case, the points are in the same level-$0$ cell.
    Thus, under $\tilde\pi$ they are in a $d$-dimensional Euclidean quadtree and we can use the Z-order of $(x,z)$ and $(x',z')$ to determine which comes first.
    
    Otherwise, we need to look at the situation at the highest level $\ell$ where $p$ and $p'$ are in different cells.
    If one of the points is in the top child cell of its parent, that point comes first.
    If both are in one of the bottom child cells, then we look at which of the cells comes first in the Z-order.
    We can distinguish these cases by checking if $p$ and $p'$ are in cells of the same $z$-range at level $\ell$.
    
    Let $L(a,b) = 1 + \lfloor \log(a \oplus b) \rfloor$ give the smallest index such that all bits of index at least $L(a,b)$ in the binary representations of $a$ and $b$ match (where $\oplus$ denotes bitwise exclusive-or).
    At some level $\ell$, the points $p$ and $p'$ are in cells of the same $z$-range if $\lfloor z / 2^\ell \rfloor = \lfloor z' / 2^\ell \rfloor$.
    In other words, their binary expansions match for the bits of index at least $\ell$, meaning $L(z,z')$ is the smallest value of $\ell$ for which this holds.
    Under the $\tilde\pi_z$ projection this $z$-range starts at the nearest multiple of $2^{L(z,z')}$ below $z$, which is $\lfloor z / 2^{L(z,z')} \rfloor \cdot 2^{L(z,z')}$.
    Without the projection this becomes $2^{\lfloor z / 2^{L(z,z')} \rfloor \cdot 2^{L(z,z')}}$.
    We can now look at the cells with this $z$-range.
    As noted in Section~\ref{sec:covering}, under $\tilde\pi_x$ these cells form a $(d-1)$-dimensional Euclidean grid.
    The side length is $\sqrt{d-1}$ times the Euclidean width of the original horobox cells, which is by definition their width multiplied by their $z$-coordinate and thus
    \[ 2^{\lfloor z / 2^{L(z,z')} \rfloor \cdot 2^{L(z,z')}} \cdot \frac{2^{2^{L(z,z')}-1}}{\sqrt{d-1}}. \]
    We can rewrite this as $\frac{2^{L^*}}{\sqrt{d-1}}$ with $L^* = \lfloor z / 2^{L(z,z')} \rfloor \cdot 2^{L(z,z')} + 2^{L(z,z')} - 1$.
    Thus, $p$ and $p'$ will be in the same level $L(z,z')$ cell if for all $i=1,\dots,d-1$ we have $\lfloor x_i / 2^{L^*} \rfloor = \lfloor x_i' / 2^{L^*} \rfloor$ or equivalently $L(x_i, x_i') \leq L^*$.

    This check lets us know if $p$ and $p'$ are already in the same cell at level $L(z,z')$. If they are, then they were first split with a horosphere normal to the $z$ axis. We can thus determine their L-order by simply checking if $z < z'$.
    Otherwise, they were split first with some hyperplane normal to some $x_i$-axis for some $i\in 1,\dots,d-1$, thus we can find which of $p$ or $p'$ comes first by determining the $d-1$-dimensional Euclidean Z-order of $x$ and $x'$.
\end{proof}

\section{Steiner spanners and \texorpdfstring{$(1+\eps)$}{(1+ε)}-approximate nearest neighbours}\label{sec:approx}
Using Theorem~\ref{thm:locsens} and Lemma~\ref{lem:lorder} it is already possible to get constant-factor approximations with the same methods as in Euclidean space (see Section~\ref{sec:constapprox}).
On the other hand, we will see that better approximations require new methods.

\subsection{Lower bound on hyperbolic spanners}
In constant-dimensional Euclidean space we can get $t$-spanners with $o(n^2)$ edges for any constant $t > 1$.
However, this is already impossible for the hyperbolic plane, because (as in high-dimensional Euclidean space) we can construct arbitrarily large sets of points where the distance between any pair of points is approximately the same.

\begin{lemma}
    For any $n \geq 2$ and any $\eps \in (0,1]$, the point set $P(n,\eps) \in \Hyp^2$ of $n$ points equally spaced around a circle of radius $r = \frac{1}{\eps} \ln n$ has the property that any distance between a pair of points is in $(2r (1-\eps), 2r]$.
\end{lemma}
\begin{proof}
    To prove this, we can use the hyperbolic law of sines:
    for a triangle with sides $a,b,c$ and angles $A,B,C$ opposite the respective sides,
    $\frac{\sin A}{\sinh a} = \frac{\sin B}{\sinh b} = \frac{\sin C}{\sinh c}$.
    We look at the triangle formed by the circle centre, a point $p \in P(n, \eps)$ and the point $q$ halfway along the geodesic from $p$ to a neighbour.
    If we pick $a = \distHd(p,q)$ and $b = r$, this gives $\frac{\sin\frac{\pi}{n}}{\sinh a} = \frac{1}{\sinh r}$ and thus $a = \arsinh(\sinh r \cdot \sin\frac{\pi}{n})$.
    Because $r \geq \ln 2$, we can say $\sinh r \geq \frac{3}{8} e^r$ and in general $\arsinh x > \ln 2x$, so $a > r + \ln\left( \frac{3}{4}\sin\frac{\pi}{n} \right)$.
    Furthermore, $\sin\frac{\pi}{n} \geq \frac{2}{n}$ for $n \geq 2$, which finally gives $a > r + \ln\frac{3}{2n} = r - \eps r + \ln\frac{3}{2} > (1 - \eps)r$.
    Therefore, the distance between $p$ and a neighbour is greater than $(1 - \eps)2r$.
    The maximum distance comes from the circle's diameter.
\end{proof}

These point sets show that it is often not possible to approximate distances with the same techniques as in constant-dimensional Euclidean space.
Low-stretch spanners, and generally any technique that can induce such spanners, will fail.
One such technique is the \emph{well-separated pair decomposition} \cite{CallahanK95} (WSPD).
A pair of point sets $A,B$ is \emph{$s$-separated} when $s \cdot \max\{ \diam(A), \diam(B) \} \leq \dist(A, B)$.
The well-separated pair decomposition of set $P$ is then a collection of $s$-separated pairs such that for any $p,q \in P$ there is a pair $A,B$ with $p \in A$ and $q \in B$.
\emph{Locality-sensitive orderings} \cite{doi:10.1137/19M1246493} also induce a spanner.
Given $\eps \in (0,1/2]$, these are orderings such that for any two points $p,q$, there is an ordering where any point $u$ between $p$ and $q$ in the ordering must be within distance $\eps \dist(p,q)$ from $p$ or $q$.

\begin{corollary}\label{cor:spanner_lower}
    There are point sets in $\Hyp^d$ such that:
    \begin{itemize}
        \item No $t$-spanner exists for $t < 2$ other than the complete graph.
        \item Any well-separated pair decomposition has size $\bigOm{n^2}$.
        \item The number of locality-sensitive orderings needed is always $\bigOm{n}$.
    \end{itemize}
\end{corollary}
\begin{proof}
    Any pair of points in $P(n,\eps)$ not connected directly in the graph will have graph distance more than $\frac{2(1-\eps) 2r}{2r} = 2-2\eps$ times larger than their hyperbolic distance, meaning any $t$-spanner with $t \leq 2-2\eps$ needs to connect all pairs of points.
    The only way a pair $A,B \subseteq P(n,1-\frac{1}{s})$ can be $s$-separated is when both contain only a single element.
    If for example $|A| \geq 2$, then $s \cdot \diam(A) > 2r \geq \distHd(A, B)$.
    Thus, any WSPD must have size $\bigOm{n^2}$.

    If two points $p,q \in P(n,1-\eps)$ are not adjacent in any locality-sensitive ordering, then there is a point $u$ between them where $\min\{ \distHd(p,u),\distHd(q,u) \} > 2r \eps$ while also $\distHd(p,q) \leq 2r$, which makes it impossible to get $\min\{ \distHd(p,u),\distHd(q,u) \} \leq \eps \distHd(p,q)$ as required.
    Thus, any pair needs to be adjacent in some ordering.
    This requires $\bigOm{n}$ orderings.
\end{proof}

We will get around this lower bound by using \emph{Steiner points} and also match it by giving sparse $t$-spanners for any $t > 2$ (leaving the case $t=2$ open).

\subsection{Steiner spanners}
If we add the centre of the circle to the point set $P(n,\eps)$, suddenly we can make sparse spanners.
These are \emph{Steiner spanners} for the original point set:
a Steiner $t$-spanner for a point set $P$ is a geometric graph on $P \cup S$ that is a $t$-spanner for the points of $P$, where $S$ are the \emph{Steiner points}.
Theorem~\ref{thm:steinerspanner} gives a result for Steiner spanners similar to that given by locality-sensitive orderings for spanners in Euclidean space.

\steinerspanner*

Because this Steiner spanner has a nice (bipartite) structure, we can also use it to immediately get a result for normal spanners.

\begin{corollary}
    We can construct in $n \log n \cdot d^{\bigO{d}}\log(1/\eps)/\eps^d$ time a $(2+\eps)$-spanner that has at most $n \cdot d^{\bigO{d}}\log(1/\eps)/\eps^d$ edges.
\end{corollary}
\begin{proof}
    Start by constructing a Steiner $(1+\eps/2)$-spanner according to Theorem~\ref{thm:steinerspanner}.
    To get a normal spanner, we have to remove the Steiner points.
    For each Steiner point $s$, find the closest point $q_s$ that is connected to it.
    Now, connect each point that was connected to $s$ to $q_s$ instead.
    This increases the graph distance between any two points connected to $s$ by at most $2\distHd(s,q_s)$, which means their distance at most doubles.
    Thus, overall the graph distances at most double as well, giving a $(2+\eps)$-spanner.
\end{proof}

The remainder of this section will focus on proving Theorem~\ref{thm:steinerspanner}.
First, we need two properties that significantly limit where the geodesics through a given point can go.

A \emph{cylinder} of radius $r$ around a given geodesic $\ell$ is the set of points with distance at most $r$ to $\ell$.
We will use $\cyl(p)$ to denote the (infinite) cylinder of radius $\arsinh(1)$ around the vertical line through $p$, i.e.\ $\cyl(x,z) = \{ (x',z') \in \Hyp^d \mid \|x - x'\| \leq z' \}$.
We will also call this the \emph{vertical unit cylinder} of $(x,z)$.
This cylinder appears as a Euclidean cone with its apex at $(x,0)$ and aperture $\frac{\pi}{2}$.

\begin{lemma}\label{lem:cylinders}
    For any $p,q \in \Hyp^d$, the geodesic between them will go through $\cyl(p) \cap \cyl(q)$.
\end{lemma}
\begin{proof}
    The geodesic between $p$ and $q$ is an arc from some Euclidean circle normal to the hyperplane $z=0$.
    Using $T_{\sigma,\tau}$ we can always isometrically transform this circle to a unit circle centred at the origin, so without loss of generality we will only look at that case.
    Now, let let $t$ be the highest point on the geodesic. Notice that $\|x(t)\|\leq 1$.
    Then,
    \[ \|x(p) - x(t)\|^2 < (1 - \|x(t)\|)^2 \leq 1 - \|x(t)\|^2 = z(t)^2. \]
    Therefore $t \in \cyl(p)$ and we can similarly argue that $t \in \cyl(q)$.
\end{proof}
The cells of a hyperbolic quadtree are not convex, unlike their Euclidean counterparts.
However, they are still \emph{star-shaped}:
a cell $C$ has a non-empty subset $K \subseteq C$ (its \emph{kernel}), such that any geodesic between a point in $K$ and a point in $C$ is fully contained in $C$.
\begin{lemma}\label{lem:starshape}
    For any $p,q \in \Hyp^d$ and any hyperbolic quadtree cell $C$ that contains both, if $p$ lies below $z = \zd(C) \cdot 3^{h(C)/2}$ then the geodesic between $p$ and $q$ is completely contained in $C$.
\end{lemma}
\begin{proof}
    Let $\partial C^\up$ denote the upper boundary horosphere of horobox $C$.
    For any point $t \in \partial C^\up$, we define the halfspace $H_t$ as the set of points below the hyperplane tangent to $\partial C^\up$ at $t$.
    We will first show that $K = \bigcap_{t \in \partial C^\up} H_t \cap C$ lies in the kernel of $C$.

    Given $p \in K$ and $q \in C$ we need to show that their geodesic $pq$ is contained in $C$.
    For this, let $t$ be the point on $\partial C^\up$ with $x(t) = x(q)$.
    Now, $p,q \in H_t \cap C$ and we can show $pq \subset H_t \cap C$.
    Because $pq$ is a Euclidean circle arc between $p$ and $q$, it cannot intersect the lower boundary or the sides of $C$.
    In the Beltrami-Klein model of hyperbolic space, $H_t$ shows up as a Euclidean halfspace and $pq$ as a Euclidean line segment, which must therefore be fully contained in $H_t$.
    Thus, $pq \subset H_t \cap C$ and $K$ lies in the kernel of $C$.

    We will now consider what $K$ looks like for specific width and height of $C$.
    Each halfspace $H_t$ appears as a Euclidean ball centred at $(x(t), 0)$ with radius $z(t)$.
    For any point on the upper boundary of $K$, its $z$-value is defined by the ball whose point $t \in \partial C^\up$ is furthest away.
    Thus, the lowest $z$-value on the upper boundary of $K$ is given by the $z$-value at which $H_{(\x(C), \zu(C))}$ intersects the line $x = \xm(C)$ opposite it.
    This intersection happens at $z = \sqrt{\zu(C)^2 - \|\xm(C) - \x(C)\|^2} = \zd(C) \sqrt{4^{h(C)} - (d-1) w(C)^2}$.
    When $C$ is a quadtree cell at level $\ell \geq 0$ this evaluates to $z = \zu(C) \cdot \frac{\sqrt{3}}{2} \geq \zd(C) \cdot 3^{h(C)/2}$.
    For $\ell \leq -1$ we get $z = \zd(C) \cdot \sqrt{4^{2^\ell} - \alpha^2 4^\ell} > \zd(C) \cdot \sqrt{3^{2^\ell}} = \zd(C) \cdot 3^{h(C)/2}$.
\end{proof}

\subparagraph*{Spanner construction.}
For each of the $3d+3$ infinite hyperbolic quadtrees from Theorem~\ref{thm:locsens}, we sort $P$ based on the corresponding L-order.
Then, for each pair of points $p,q$ adjacent in an L-order, we find the cell $C$ at the lowest level $\ell$ that contains both $p$ and $q$ in the corresponding infinite hyperbolic quadtree.
Let $c$ be the constant hidden by the big-O notation in Theorem~\ref{thm:locsens}.
We apply quadtree splits to $C$ until we get a set of cells we will call $\cheps(C)$, where each element has diameter at most $\frac{\eps}{c \cdot d \sqrt d}$ times that of $C$.
By Theorem~\ref{thm:quadtree}(ii), this requires $\log( c \cdot d \sqrt d / \eps) + \bigTh{1}$ quadtree splits.
Both $p$ and $q$ are then connected to a Steiner point in each cell of $\cheps(C)$ that intersects $\cyl(p)$, respectively $\cyl(q)$.
We repeat this procedure for all ancestor cells $C'$ of $C$ where there is a $\hat C\in \cheps(C')$ such that $\hat C \subseteq C \subseteq C'$.
Note that this gives a bipartite graph:
all edges are between an input point and a Steiner point.

\begin{claim}\label{claim:steinerdegree}
    A point $p \in P$ gets connected to $2^{\bigO d} d^{\frac{3}{2}d} \log(d / \eps) / \eps^d$ Steiner points.
    When the distance $\delta$ from $p$ to its nearest neighbour in $P$ is large enough this improves to
    \begin{itemize}
        \item $\frac{2^{\bigO d} d^{2d} \log(d / \eps)}{\delta^{d-1} \eps^d}$ Steiner points  if $\sqrt{d} < \delta < d^2 / \eps$,
        \item $\frac{2^{\bigO d} \log(d / \eps)}{\eps}$ Steiner points if $\delta \geq d^2 / \eps$.
    \end{itemize}
\end{claim}
\begin{proof}
    For $p$, the procedure above happens in $\bigTh{d \log\left( d / \eps \right)}$ cells $C'$.
    By Lemma~\ref{lem:cells} the cells in $\cheps(C')$ have width $\bigOm{\eps \delta / d^2}$.
    To count how many cells of $\cheps(C')$ intersect $\cyl(p)$, we consider at any given $z$ how many cells intersect a $(d-1)$-dimensional ball normal to the $z$-axis with Euclidean radius $z$.
    The number of cells of width $w$ that intersect such a ball can be bounded by $\max\{2, 4/w + 1\}^{d-1} = \max\left\{ 2^{d-1},\ \bigO{ \frac{d^2}{\eps \delta} }^{d-1} \right\}$.
    There are $\bigTh{d \sqrt d / \eps}$ distinct $z$-coordinates in $\cheps(C')$, so in total this becomes $\max\left\{ 2^d \cdot \bigO{d \sqrt d / \eps},\ 2^{\bigO d} \cdot  \frac{d^{2d-\frac{1}{2}}}{\eps^d \delta^{d-1}} \right\}$.

    For $\ell(C') \leq 0$, we get a better bound by noticing $|\cheps(C')| = \bigTh{\left( c \cdot d \sqrt d / \eps \right)^d}$.
    We can further generalise this to $\ell(C') \leq \log d$, because then $C'$ is composed of up to $2^d - 1$ level-0 cells and therefore $|\cheps(C')| = \bigO{\left( 2c \cdot d \sqrt d / \eps \right)^d}$.
\end{proof}

\begin{claim}
    This construction gives a Steiner $(1+7\eps)$-spanner when $\eps \leq \frac{1}{14}$.
\end{claim}
\begin{proof}
    Let $\dist_G$ denote the distance in the graph.
    Given $p,q \in P$ we want to prove $\dist_G(p,q) \leq (1+7\eps) \distHd(p,q)$.
    From Theorem~\ref{thm:locsens} we know that $p$ and $q$ are contained in a cell $C$ in one of the infinite hyperbolic quadtrees where $\diam(C) = \bigO{d \sqrt d} \distHd(p,q)$.
    By construction, the cells of $\cheps(C)$ have diameter at most $\eps \distHd(p,q)$.
    Let $p' \in P$ be the point in the same cell from $\cheps(C)$ as $p$ closest in the L-order to $q$, and $q'$ defined analogously.
    Now, $p'$ and $q'$ must both be connected to a Steiner point in each cell of $\cheps(C)$ that intersects $\cyl(p')$, respectively $\cyl(q')$.
    By the combination of Lemma~\ref{lem:cylinders} and Lemma~\ref{lem:starshape}, this means that one of the Steiner points they are both connected to lies in a cell of $\cheps(C)$ that is intersected by the geodesic between $p'$ and $q'$.
    Therefore, $\dist_G(p',q') \leq \distHd(p',q') + 2\eps\distHd(p,q)$.
    We can get $\dist_G(p',q') \leq (1 + 4\eps) \distHd(p,q)$ from this by noticing that $\distHd(p',q') \leq \distHd(p,p') + \distHd(p,q) + \distHd(q,q') \leq (1 + 2\eps) \distHd(p,q)$.

    We will now use induction to prove that $\dist_G(p,q) \leq (1 + 7\eps) \distHd(p,q)$.
    For $\distHd(p,q) = 0$ this holds trivially.
    As induction hypothesis, we now assume that for any pair $v,w \in P$ closer together than $p$ and $q$, we have $\dist_G(v,w) \leq (1+7\eps) \distHd(v,w)$.
    In particular, this gives us $\dist_G(p,p') \leq (1+7\eps) \distHd(p,p')$.
    Because $p$ and $p'$ lie in a cell of diameter at most $\eps \distHd(p,q)$ and $\eps \leq \frac{1}{14}$, we have $\dist_G(p,p') \leq (1 + 7\eps) \eps\distHd(p,q) \leq \frac{3}{2}\eps\distHd(p,q)$.
    For the same reason, $\dist_G(q,q') \leq \frac{3}{2}\eps\distHd(p,q)$ and thus $\dist_G(p,q) \leq (1 + 7\eps)\distHd(p,q)$.
\end{proof}

After replacing $\eps$ by $\eps/7$, these two claims together prove most of Theorem~\ref{thm:steinerspanner}.
What remains is to show that we can efficiently maintain the Steiner spanner while inserting and deleting points.

\subparagraph*{Dynamic manipulation and wrap-up of the proof of Theorem~\ref{thm:steinerspanner}.}
To maintain the construction dynamically, we only need self-balancing binary search trees (e.g.\ red-black tree~\cite{CormenLRS01}).
For each L-order $i = 1, \dots, 3d+3$ we maintain the points of $P$ in that order using a balanced binary search tree $P_i$.
Additionally, we maintain a balanced binary search tree $S$ containing the Steiner points that functions as an associative array:
each Steiner point $s \in S$ is associated with a set of points $E_s$, containing the input points connected to it.
In other words, these represent the edge set of the Steiner spanner.
These sets can again be implemented as self-balancing binary search trees.
Each binary tree for $P_i$, $S$ or $E_s$ will contain less than $n^2$ points\footnote{If $\eps$ is small enough to lead to $\Omega(n^2)$ Steiner points, we can use the complete graph on $P$ as spanner to get the claimed bounds.} and has comparisons that take $\bigO{d}$ time, so searching, adding and removing take $\bigO{d \log n}$ time.

For a point $p \in P$, we can determine in $\bigO{d^2\log n} + d^{\bigO{d}}\log(1/\eps)/\eps^d$ time which Steiner points it should be connected to by the following procedure.
For each $i = 1, \dots, 3d+3$, find (in $\bigO{d \log n}$ time) its neighbours in $P_i$.
Then, for each neighbour we consider the smallest cell $C$ in the $i^\text{th}$ infinite hyperbolic quadtree that it shares with $p$.
The cells to connect to are those in $\cheps(C')$
that intersect $\cyl(p)$, where $C'$ starts as $C$ and goes up its ancestors until $C \in \cheps(C')$.
For each cell, we take its centre as its corresponding Steiner point.
Enumerating all these Steiner points takes $d^{\bigO{d}}\log(1/\eps)/\eps^d$ time.

Finding a Steiner point $s$ in $S$ takes $\bigO{d \log n}$ time, as does adding/removing a connection in $E_s$ (or adding/removing $E_s$ itself, if it did not exist yet or has no connections left).

Adding and removing a point $p$ only affects the points adjacent to $p$ in the L-orders and $p$ itself, so $\bigO{d}$ points altogether.
For each of the affected points, we can remove all connections to Steiner points it had before and then add connections to Steiner points in the new configuration.
This can affect their in total $d^{\bigO{d}}\log(1/\eps)/\eps^d$ connections to Steiner points, so the update takes $\log n \cdot d^{\bigO{d}}\log(1/\eps)/\eps^d$ time using the operations from the previous paragraph.

\subsection{Approximate Nearest Neighbours}
Using Theorem~\ref{thm:steinerspanner} we can also get a data structure for dynamic $(1+\eps)$-approximate nearest neighbours.
We only need to modify the data structure slightly:
for each Steiner point $s$, the points in $E_s$ will now be sorted based on their distance to $s$.
To query for the nearest neighbour of $p$, we then find all Steiner points $p$ would get connected to if it were to get inserted into the data structure.
For each of these Steiner points, we retrieve the input point closest to it.
This gives a small set of candidate points, from which we return the point closest to $p$.
The returned point is the exact nearest neighbour to $p$ using the spanner distances, thus it will be an $(1+\eps)$-approximate nearest neighbour in the actual space.

The sets $E_s$ together store all the edges of the Steiner spanner exactly once and $S$ contains the ($d$-dimensional) Steiner points, so these together take $n \cdot d^{\bigO{d}}\log(1/\eps)/\eps^d$ space.
The size of the trees $P_i$ sums up to $\bigO{d^2n}$, thus the final data structure still requires $n \cdot d^{\bigO{d}}\log(1/\eps)/\eps^d$ space. 

\approxnn*

This method can also be used to dynamically maintain an approximate (bichromatic) closest pair, as done in \cite{doi:10.1137/19M1246493}.
Given two sets $R$ (red) and $B$ (blue) of points, a bichromatic closest pair is a pair $(r,b) \in R \times B$ such that the distance $\distHd(r,b)$ is minimal.
A pair $(r',b')$ is a $(1+\eps)$-approximate bichromatic closest pair when $\distHd(r',b') \leq (1+\eps) \distHd(r,b)$.
To dynamically maintain such a pair, each Steiner point $s$ will need to maintain two sets $R_s$ and $B_s$ instead of the single set $E_s$, corresponding to the red points and the blue points connected to $s$.
The points in both are again sorted by distance to $s$.
Each Steiner point $s$ now gives its closest red and blue point as candidate for a $(1+\eps)$-approximate bichromatic closest pair.

Using dynamically-maintained approximate bichromatic closest pairs, we can also compute approximate minimum spanning trees~\cite{Indyk98} and approximate minimum bottleneck matchings~\cite{Goel01}.

\section{Constant-approximate nearest neighbours}\label{sec:constapprox}
Theorem~\ref{thm:locsens} and Lemma~\ref{lem:lorder} are equivalent to statements about Euclidean quadtrees Chan~\etal~\cite{doi:10.1137/19M1246493} use to find an approximate nearest neighbour and bichromatic closest pair, so we can do the same in hyperbolic space. 
For both problems the algorithm is similar.
Given a point set $P$, we first make $3d+3$ self-balancing binary search trees (e.g.\ red-black tree~\cite{CormenLRS01}) where each sorts the points based on the L-order from one of the infinite quadtrees from Theorem~\ref{thm:locsens}.
By Lemma~\ref{lem:lorder} this takes $\bigO{d^2 n \log n}$ time. Notice that we can add or remove a point in all of these trees in $\bigO{d^2 \log n}$ time.

To get the nearest neighbour of some point $p$, we determine where it would end up in the L-order for each of the trees, then return $q$; the closest of the neighbours.
This takes $\bigO{d^2 \log n}$ time and gives an $\bigO{d \sqrt d}$-approximate nearest neighbour:
the actual nearest neighbour $q'$ of $p$ will be in the smallest cell $p$ is in.
Because the returned point $q$ was closer to $p$ in the L-order than $q'$, it must also be in that cell.
By Theorem~\ref{thm:locsens} this means it can be at most $\bigO{d \sqrt d}$ times further away, making it an $\bigO{d \sqrt d}$-approximate nearest neighbour.

The same reasoning can be used to get an $\bigO{d \sqrt d}$-approximate bichromatic closest pair.
For this, we go through the sorted lists and return the pair of neighbouring points of different colour that is closest.
This takes $\bigO{d^2 n}$ time. The above reasoning yields the following theorem.

\begin{theorem} \label{thm:applications}
Let $P \subset \Hyp^d$ be a given set of $n$ points.
\begin{itemize}
\item We can find an $\bigO{d \sqrt d}$-approximate bichromatic closest pair of $P$ in $\bigO{d^2 n \log n}$ time.
\item We can construct a data structure in $\bigO{d^2 n \log n}$ time that uses $\bigO{d^2 n}$ space\footnote{By only storing the full points once and then referring to them with pointers this can be reduced to $\bigO{dn}$ space.} and can answer queries for an $\bigO{d \sqrt d}$-approximate nearest neighbour in $P$ in $\bigO{d^2 \log n}$ time, and perform updates (point insertions and removals) in $\bigO{d^2 \log n}$ time.
\end{itemize}
\end{theorem}

\section{Proof of Theorem~\ref{thm:quadtree}}\label{sec:quadtreeproof}
\quadtree*
\begin{proof}
\begin{enumerate}[(i)]
    \item
    First assume $C$ is a cell at level $\ell \geq 4$ with child $C'$.
    Then, $\frac{\diam(C')}{\diam(C)} = \frac{2\arsinh\left(2^{2^{\ell-1}-2} \right)}{2\arsinh\left(2^{2^\ell-2} \right)}$ by Lemma~\ref{lem:cells}.
    Here we can use $\ln 2x < \arsinh x < \ln 4x$ to get
    \begin{align*}
        \frac{\arsinh\left(2^{2^{\ell-1}-2} \right)}{\arsinh\left(2^{2^\ell-2} \right)}
        < \frac{\ln\left(4 \cdot 2^{2^{\ell-1}-2} \right)}{\ln\left(2 \cdot 2^{2^\ell-2} \right)}
        = \frac{2^{\ell-1}}{2^\ell - 1}
        = \frac{1}{2} + \frac{1}{2^{\ell + 1} - 2}
        \leq \frac{8}{15},
        \\
        \frac{\arsinh\left(2^{2^{\ell-1}-2} \right)}{\arsinh\left(2^{2^\ell-2} \right)}
        > \frac{\ln\left(2 \cdot 2^{2^{\ell-1}-2} \right)}{\ln\left(4 \cdot 2^{2^\ell-2} \right)}
        = \frac{2^{\ell-1} - 1}{2^\ell}
        = \frac{1}{2} - 2^{-\ell}
        \geq \frac{7}{16}.
    \end{align*}

    Now assume $C$ is at level $\ell \leq -2$.
    Then by Lemma~\ref{lem:cells},
    \[\frac{\diam(C')}{\diam(C)} = \frac{2\arsinh\left(\frac12 \sqrt{\frac{(\alpha')^2 \cdot 4^{\ell-1} + (2^{2^{\ell-1}}-1)^2}{2^{2^{\ell-1}}}} \right)}{2\arsinh\left(\frac12 \sqrt{\frac{\alpha^2 \cdot 4^\ell + (2^{2^\ell}-1)^2}{2^{2^\ell}}} \right)},\]
    where $\alpha$ corresponds to $C$ and $\alpha'$ to $C'$.
    From Lemma~\ref{lem:cells} it follows that either $\alpha'=\alpha$ or $\alpha'=\alpha \cdot 2^{-2^{\ell-1}}$.
    First, we look for an upper bound.
    We observe that in the denominator, the argument $x$ of $\arsinh x$ is between $0$ and $0.15$, thus we may assume $\arsinh x > 0.997 x$.
    For the positive numerator we can use that $\arsinh x \leq x$.
    This gives
    \[
        \frac{2\arsinh\left(\frac12 \sqrt{\frac{(\alpha')^2 \cdot 4^{\ell-1} + (2^{2^{\ell-1}}-1)^2}{2^{2^{\ell-1}}}} \right)}{2\arsinh\left(\frac12 \sqrt{\frac{\alpha^2 \cdot 4^\ell + (2^{2^\ell}-1)^2}{2^{2^\ell}}} \right)}
        <
        \frac{2^{2^{\ell-2}}}{0.997} \sqrt{\frac{ (\alpha')^2 \cdot 4^{\ell-1} + (2^{2^{\ell-1}}-1)^2}{ \alpha^2 \cdot 4^\ell + (2^{2^\ell}-1)^2}}.
    \]
    We can also use that $\alpha' \leq \alpha$, that $(2^{2^\ell}-1)^2 \geq 0$ and that $(2^{2^{\ell-1}}-1)^2 < 0.07 \cdot 4^\ell$ for $\ell \leq -2$, giving
    \[
        \frac{2^{2^{\ell-2}}}{0.997} \sqrt{\frac{ (\alpha')^2 \cdot 4^{\ell-1} + (2^{2^{\ell-1}}-1)^2}{ \alpha^2 \cdot 4^\ell + (2^{2^\ell}-1)^2}}
        <
        \frac{2^{2^{\ell-2}}}{0.997} \sqrt{\frac{1.07 \cdot 4^{\ell-1}}{4^\ell}}
        <
        0.55.
    \]

    To prove a lower bound we work similarly.
    We now additionally use that $\alpha' \geq 2^{-2^{\ell-1}} \alpha$ and that $(2^{2^\ell}-1)^2 < 0.15 \cdot 4^\ell$ for $\ell \leq -2$, giving
    \begin{align*}
        \frac{2\arsinh\left(\frac12 \sqrt{\frac{(\alpha')^2 \cdot 4^{\ell-1} + (2^{2^{\ell-1}}-1)^2}{2^{2^{\ell-1}}}} \right)}{2\arsinh\left(\frac12 \sqrt{\frac{\alpha^2 \cdot 4^\ell + (2^{2^\ell}-1)^2}{2^{2^\ell}}} \right)}
        &>
        0.997 \cdot 2^{2^{\ell-2}} \sqrt{\frac{ (\alpha')^2 \cdot 4^{\ell-1} + (2^{2^{\ell-1}}-1)^2}{ \alpha^2 \cdot 4^\ell + (2^{2^\ell}-1)^2}} \\
        &\geq
        0.997 \sqrt{\frac{4^{\ell-1}}{ 4^\ell + (2^{2^\ell}-1)^2}} \\
        &>
        0.997 \sqrt{\frac{4^{\ell-1}}{ 1.15 \cdot 4^\ell}} \\
        &>
        0.46.
    \end{align*}

    This proves the statement for $\ell \leq -2$ and $\ell \geq 4$.
    We check the remaining cases in Table~\ref{tab:my_label}.
    The lower bound comes from $\ell=2$ and the upper bound from $\ell=0$ with $\alpha'=1$.

    \begin{table}[ht]
        \centering
        \begin{tabular}{|c|c|c|c|c|c|c|c|c|c|}
            \hline
            $\ell$ & \multicolumn{4}{c|}{-1} & \multicolumn{2}{c|}{0} & 1 & 2 & 3
            \\\cline{1-7}
            $\alpha$  & \multicolumn{2}{c|}{$1/\sqrt{2}$} & \multicolumn{2}{c|}{$1$} & \multicolumn{2}{c|}{1} &&&
            \\\cline{1-7}
            $\alpha'$ & $1/\sqrt[4]{8}$ & $1/\sqrt{2}$ & $1/\sqrt[4]{2}$ & $1$ & $1/\sqrt{2}$ & $1$ &&&
            \\\hline\hline
            Ratio & 0.485 & 0.5218 & 0.4795 & 0.5312 & 0.4718 & 0.5605 & 0.526 & 0.4208 & 0.4317
            \\\hline
        \end{tabular}
        \caption{Ratios between the diameter of a child cell and the diameter of its level $\ell$ parent, up to four decimal places.}
        \label{tab:my_label}
    \end{table}

    \item
    First, let $\ell \geq 0$.
    Then by Lemma~\ref{lem:cells}, $\diam(C) = 2\arsinh\left(2^{2^\ell-2} \right)$.
    Here the argument of $\arsinh$ is always at least $\frac{1}{2}$, so we can say $\arsinh x = \bigTh{\log x}$ and get
    \[
        2\arsinh\left(2^{2^\ell-2} \right)
        = \bigTh{\log\left(2^{2^\ell-2} \right)}
        = \bigTh{2^\ell}.
    \]
    Now, let $\ell < 0$.
    By Lemma~\ref{lem:cells}, $\diam(C) = 2\arsinh\left(\frac12 \sqrt{\frac{\alpha^2 \cdot 4^{-k} + (2^{2^{-k}}-1)^2}{2^{2^{-k}}}} \right)$ for some constant $\alpha$.
    Here the argument of $\arsinh$ is always less than $\frac{1}{2}$, so we can say $\arsinh x = \bigTh{x}$ and get
    \[
        2\arsinh\left(\frac12 \sqrt{\frac{\alpha^2 \cdot 4^\ell + (2^{2^\ell}-1)^2}{2^{2^\ell}}} \right)
        = \bigTh{\sqrt{\frac{\alpha^2 \cdot 4^\ell + (2^{2^\ell}-1)^2}{2^{2^\ell}}}}
        = \bigTh{2^\ell}. \qedhere
    \]
    
    \item
    \begin{figure}[b]
        \centering
        \includegraphics{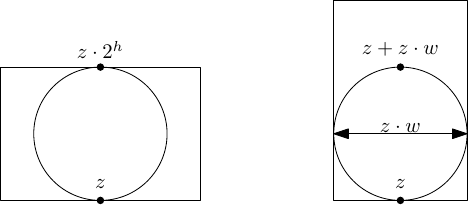}
        \caption{The two possible situations for the inscribed ball of a cube-based horobox, with last coordinates indicated as well as the Euclidean width of the rectangle in the second case.}
        \label{fig:inscribed_ball}
    \end{figure}

    For a cell $C$, the inscribed ball will touch the bounding horospheres at the top and bottom, or the bounding hyperplanes at the sides; see Figure~\ref{fig:inscribed_ball}.
    The hyperbolic diameter of a ball is given by the distance between its highest and lowest point.
    In the first case this is $\ln\left( \frac{\zu(C)}{\zd(C)} \right) = h(C) \ln 2$.
    In the second case, we only know the ball's Euclidean diameter is $\zd(C) \cdot w(C)$, but from that we can calculate that the highest point of the ball is at height $\zd(C) + \zd(C) \cdot w(C)$, so the distance is $\ln\left( \frac{\zd(C) + \zd(C) \cdot w(C)}{\zd(C)} \right) = \ln(w(C) + 1)$.

    In general the diameter of the inscribed ball is the smallest of these two.
    We will show that for quadtree cells this is always $\bigOm{\ln(w(C)+1)}$.
    At level $\ell$, the first case gives $h(C) \ln 2 = \ln\left( 2^{2^\ell} \right)$ for every cell.
    In the second case, for $\ell \geq 0$ we have $\ln(w(C)+1) = \ln\left( \frac{2^{2^\ell-1}}{\sqrt{d-1}} + 1 \right)$, which is at most $\ln\left( 2^{2^\ell} + 1 - \frac{1}{2} \cdot 2^{2^\ell} \right)$ and thus smaller.
    For $\ell < 0$ we have $\ln(w(C)+1) \leq \ln\left( \frac{1 \cdot 2^\ell}{\sqrt{d-1}} + 1 \right)$ by Lemma~\ref{lem:cells}, which is at most $\frac{2^\ell}{\sqrt{d-1}}$ and thus $h(C) \ln 2 = \bigOm{\ln(w(C)+1)}$.
    
    We need to divide the diameter of the inscribed ball by the diameter of the circumscribed ball to get the fatness.
    Notice that $2\diam(C)$ is an upper bound on the diameter of the circumscribed ball, where (ii) tells us that $\diam(C) = \bigTh{2^\ell}$.
    First assume $w < 1$ and $\ell < 0$.
    Then the fatness is at least
    \[
        \frac{ \bigOm{\ln(w(C) + 1)} }{\bigO{2^\ell}}
        = \bigOm{ \frac{w(C)}{2^\ell} }
        = \bigOm{ \frac{2^\ell}{2^\ell \sqrt d} }
        = \bigOm{\frac{1}{\sqrt d}}.
    \]
    Now assume $\ell \geq 0$ instead.
    \[
        \frac{ \bigOm{\ln(w(C) + 1)} }{\bigO{2^\ell}}
        = \bigOm{ \frac{w(C)}{2^\ell} }
        = \bigOm{ \frac{2^{2^\ell}}{2^\ell \sqrt d} }
        = \bigOm{\frac{1}{\sqrt d}}.
    \]
    Finally, assume $w \geq 1$, which means automatically means $\ell \geq 0$ and $2^{2^\ell-1} \geq \sqrt{d-1}$.
    \[
        \frac{ \bigOm{\ln(w(C) + 1)} }{\bigO{2^\ell}}
        = \bigOm{ \frac{ \log w(C) }{2^\ell} }
        = \bigOm{ \frac{2^\ell + \log \frac{1}{d}}{2^\ell} }
        = \bigOm{1}.
    \]
    
    \item
    Consider a quadtree cell $C$ of level $\ell$.
    If $\ell \leq 0$ then it has $2^d$ children, so we now consider $\ell > 0$.
    There we have $\diam(C) = 2\arsinh(2^{2^\ell-2})$ by Lemma~\ref{lem:cells}, which means $2^{2^\ell} = 4\sinh(\frac12\diam(C))$.
    The number of children of $C$ is
    \[
        2^{\frac{2^\ell}{2} (d-1)} + 1
        = \left( 4\sinh\left( \frac12\diam(C) \right) \right)^{\frac12(d-1)} + 1
        = 2^{\bigO{d\cdot \diam(C)}}.
    \]

    Now we will prove the result for the root of the quadtree.
    Assume this root lies at level $\ell > 0$, because otherwise it has $2^d$ children.
    By construction, we know that no quadtree cell $C$ of level $\ell - 1$ can cover $P$.
    In particular this means that the inscribed ball $B$ of $C$ cannot cover $P$, meaning $\diam(P) > \diam(B)$.
    In (iii) we showed $\diam(B) = \ln\left( \frac{2^{2^{\ell-1}-1}}{\sqrt{d-1}} + 1 \right)$, so $2^{2^{\ell-1}} = 2 \sqrt{d-1} (e^{\diam(B)} - 1)$.
    This means the number of children under the root is
    \[
        2^{\frac{2^\ell}{2} (d-1)} + 1
        = \left( 2\sqrt{d-1} \left( e^{\diam(B)} - 1 \right) \right)^{d-1} + 1
        = d^{\bigO{d\cdot \diam(P)}}
        .
    \]
    
    \item
    This follows from Lemma~\ref{lem:cells}.   \qedhere 
\end{enumerate}
\end{proof}

\section{Conclusion}\label{sec:conclusion}
We have presented a quadtree and a dynamically maintained Steiner spanner that work in hyperbolic space.
The constructions are relatively simple and closely correspond to Euclidean versions, so we hope that this will spark further research on algorithms in hyperbolic space.
As an example of this, we were able to give similar results for approximate nearest neighbour search as are known in Euclidean space, which is the best one can hope for as a point set inside a small hyperbolic neighbourhood has a nearly Euclidean metric.
One would hope that better results are possible for point sets that are more spread out~\cite{Kisfaludi-Bak21}.
While our result does not show strong signs of improvement for spread-out point sets, we can observe improved behaviour for larger dimensions for such point sets: in Claim~\ref{claim:steinerdegree} and near the end of the proofs of Theorem~\ref{thm:locsens} and Theorem~\ref{thm:quadtree}(iii), it can be seen that the dependence on $d$ is milder for longer distances.
This seems to be typical for hyperbolic space and is in line with the strong results on dimension reductions in hyperbolic spaces by Benjamini and Makarychev~\cite{benjamini2009dimension}.
Thus, an open question is how well Steiner spanners and approximate nearest neighbour data structures designed specifically for spread-out point sets can perform.

For both spanners and approximate nearest neighbours, our results are optimal in $n$ and match those provided by locality sensitive orderings~\cite{doi:10.1137/19M1246493} in Euclidean space, but there are more complex algorithms for either spanners or approximate nearest neighbours in Euclidean space that give better guarantees.
For the spanner specifically, we come close to the bound for classical spanners but do not match the improved bounds for Steiner spanners \cite{bhore2022}. Is there a hyperbolic Steiner $(1+\eps)$-spanner with only $\bigO{n / \eps^{(d-1) / 2}}$ edges?
On a related note, we do not consider spanner lightness (the total weight of the edges) and with our current construction this is unbounded. Is there a hyperbolic Steiner $(1+\eps)$-spanner whose total weight is $f(\eps,d)$ times the weight of a minimum spanning tree?
These are all avenues for further research.

\bibliography{bibliography.bib}

\end{document}